\documentclass{IEEEtran}
\usepackage{xcolor}
\usepackage{amsmath,amsfonts,amsthm}
\usepackage{amssymb}
\usepackage{mathtools}
\usepackage{algorithm}
\usepackage{algpseudocode}  
\usepackage{enumitem}

\usepackage{array}
\usepackage{subfig}
\usepackage{stfloats}
\usepackage{url}
\usepackage{verbatim}
\usepackage{graphicx}
\usepackage{cite}
\usepackage{multirow, multicol}
\usepackage[normalem]{ulem}
\usepackage{soul}
\usepackage{textcomp}

{}
{}
\newtheorem{remark}{Remark}{}
{}
\newtheorem{lemma}{Lemma}[]
\newtheorem{proposition}{Proposition}{}
\newtheorem{theorem}{Theorem}{}

\begin{document}

\title{Optimal Transport-Based Decentralized Multi-Agent Distribution Matching}

\author{Kooktae Lee,~\IEEEmembership{Member,~IEEE}
\thanks{K. Lee is with the Department of Mechanical Engineering, New Mexico Institute of Mining and Technology, Socorro, NM 87801, USA. e-mail: kooktae.lee@nmt.edu.}
\thanks{This work was supported by NSF CAREER Grant CMMI-DCSD-2145810.}}


\IEEEpubid{
}

\maketitle

\begin{abstract}
This paper presents a decentralized control framework for distribution matching in multi-agent systems (MAS), where agents collectively achieve a prescribed terminal spatial distribution. The problem is formulated using optimal transport (Wasserstein distance), which provides a 
principled measure of distributional discrepancy and serves as the basis for the control design. To avoid solving the global optimal transport problem directly, the distribution-matching objective is reformulated 
into a tractable per-agent decision process, enabling each agent to identify its desired terminal locations using only locally available information. A sequential weight-update rule is introduced to construct 
feasible local transport plans, and a memory-based correction mechanism is incorporated to maintain reliable operation under intermittent and range-limited communication. Convergence guarantees are established, 
showing cycle-wise improvement of a surrogate transport cost under both linear and nonlinear agent dynamics. Simulation results demonstrate that the proposed framework achieves effective and scalable distribution 
matching while operating fully in a decentralized manner.
\end{abstract}


\section{Introduction}
Coordinating multiple autonomous agents to achieve a desired spatial 
configuration is a central problem in multi-agent systems. Applications 
ranging from sensing and monitoring to logistics and robotic deployment 
often require agents to be arranged in patterns reflecting mission 
objectives or external specifications. Achieving such configurations 
under dynamics, communication limits, and local information constraints 
has motivated a variety of distributed strategies for shaping collective 
spatial behavior, including approaches based on coverage control, density 
regulation, and optimal transport.

\textbf{Literature Survey:}
Coverage control represents one of the most established distributed
approaches for organizing agents in space. It provides a classical
framework in which agents optimize a sensing-based coverage functional.
Foundational methods are based on Voronoi-partition gradients~\cite{cortes2004coverage}, and 
subsequent work developed adaptive distributed implementations for 
networked robots~\cite{schwager2009decentralized}. Recent extensions address heterogeneous 
environments~\cite{liu2021multi}, safety~\cite{bai2023safe}, and online learning of unknown density 
fields~\cite{zhang2024multi}. While well-suited for persistent sensing and exploration, 
coverage control converges to centroidal Voronoi configurations and 
therefore cannot enforce a prescribed terminal distribution.

A complementary direction focuses on distribution steering using 
mean-field and continuum models. Zheng et al.~\cite{zheng2021transporting} proposed 
mean-field feedback laws, but these approaches rely on global density 
information and continuum approximations. Cui and Li~\cite{cui2024density} developed 
predictive density regulation with robustness guarantees, yet still 
assume centralized density availability and typically do not provide 
decentralized finite-agent, discrete-time controllers. Survey work in 
mean-field swarm theory~\cite{elamvazhuthi2019mean} highlights these structural assumptions, 
which limit applicability under local communication and agent-level constraints.

Optimal Transport (OT)-based methods offer another viewpoint by modeling 
distribution evolution geometrically. The authors in \cite{krishnan2022multiscale} 
interpreted coverage algorithms as gradient flows in Wasserstein space, 
and their recent work proposed a distributed online OT optimizer~\cite{krishnan2025distributed}. 
Although computationally scalable, these approaches operate at the 
optimization layer and do not directly yield decentralized feedback 
controllers for agents with discrete-time dynamics. Moreover, OT-based 
density regulation typically acts on continuous probability fields rather 
than explicit finite-agent terminal configurations~\cite{chen2016optimal,bandyopadhyay2014probabilistic}.

Nonuniform area coverage methods instead align time-averaged 
trajectories with a reference density map. Ergodic metrics provide a 
trajectory-level alignment tool~\cite{mathew2011metrics, rao2024learning}, and Ivić et 
al.~\cite{ivic2019autonomous} demonstrated spatially adaptive spraying strategies. Recent 
density-driven optimal control (D$^2$OC) frameworks steer time-averaged 
densities toward reference maps while respecting agent dynamics and 
physical constraints~\cite{seo2025smcs, seo2025tcst}. Such methods shape evolving densities 
but do not specify where agents ultimately reside at the mission’s end.

While these research directions share the goal of shaping collective behavior, our approach focuses on terminal distribution matching for discrete agents. Unlike ergodicity-based methods or $D^2OC$ that optimize time-averaged trajectories, we explicitly determine finite-agent trajectories to satisfy a prescribed final configuration. This work is the first to address terminal matching within discrete agent systems, providing a clear departure from continuum-level PDE steering or trajectory alignment over time. This distinction is vital for tasks like establishing static agent formations for monitoring or coordinated multi-robot inspection, where the final static placement of each agent is more critical than its cumulative coverage behavior.

\textbf{Contributions:}
The main contributions of this work are as follows. 
1) The global Wasserstein distribution-matching problem is reformulated 
into a tractable per-agent decision process that operates using only local 
information, enabling scalable implementation without global coordination;
2) A cycle-wise convergence guarantee is established for the resulting 
distribution-matching scheme under both linear and nonlinear agent dynamics, 
providing a rigorous performance guarantee for decentralized 
optimal transport-based control;
3) A fully decentralized implementation is developed by combining local 
communication with a memory-based correction rule, ensuring robust 
operation under realistic communication limitations.
Taken together, these elements form a practical and scalable approach to 
multi-agent terminal distribution matching and establish a direct 
connection between OT-based distribution modeling and decentralized 
feedback control for finite-agent systems.

\section{Problem Formulation}
\noindent \textbf{Notation:}  
\( \mathbb{R} \) denotes the set of real numbers, and \( \mathbb{N} \) denotes the set of natural numbers. 
\( \mathbb{R}^n \) represents the \( n \)-dimensional Euclidean space. 
For a vector \( x \in \mathbb{R}^n \), \( \|x\| \) denotes the Euclidean norm, and \( x^{\top} \) denotes its transpose. 
The position of agent \( i \) at discrete time \( k \in \mathbb{N} \) is represented as \( x_i(k) \in \mathbb{R}^n \). 
The notation \( x \succeq 0 \) indicates that \( x \) satisfies the element-wise inequality, meaning each entry is non-negative. 
For any set \( \mathcal{S} \), \( |\mathcal{S}| \) denotes its cardinality (number of elements).

\medskip

Let the measures $\mu$ and $\nu$ be supported on 
point sets $\{x_i\}_{i=1}^M$ and $\{y_j\}_{j=1}^N$, where \(M\) denotes the number of agents contributing to the 
empirical distribution and \(N\) denotes the number of target sample 
points representing the prescribed distribution \(\nu\),
with corresponding weights $\{\alpha_i\}$ and $\{\beta_j\}$, respectively. In this discrete setting, the optimal transport problem between $\mu$ and $\nu$
reduces to computing the squared 2-Wasserstein distance (also known as the Kantorovich formulation \cite{villani2008optimal}),
\begin{equation}
\begin{aligned}
&\mathcal{W}_2^2(\mu, \nu) = \min_{\pi_{ij} \in \mathbb{R}^{M \times N}} \quad  \sum_{i=1}^M \sum_{j=1}^N \pi_{ij} \|x_i - y_j\|^2 \\
&\begin{array}{ll}
\text{s.t.} \quad & \sum_{j=1}^N \pi_{ij} = \alpha_i,\quad \sum_{i=1}^M \pi_{ij} = \beta_j,\quad \pi_{ij} \geq 0,\quad \forall i,j, \\
& \sum_{i=1}^M \sum_{j=1}^N \pi_{ij} = 1,
\end{array}
\end{aligned}
\label{eqn:kantorovich-lp}
\end{equation}
where $\pi_{ij}$ denotes the amount of mass transported from $x_i$ to $y_j$. The Wasserstein distance thus provides a principled way to compare two distributions, analogous to the Euclidean distance between two points.

\begin{remark}
In the discrete formulation, the coefficients $\alpha_i$ represent the 
mass or importance assigned to each agent location $x_i$. Throughout this 
work, we assume $\alpha_i = 1/M$, reflecting that each agent contributes 
equally to the empirical distribution. However, non-uniform values are 
also admissible. For example, if certain agents possess higher sensing resolution, 
stronger actuation capability, or play a more influential role in the 
team’s decision-making hierarchy, one may assign larger weights to the 
corresponding $\alpha_i$. In contrast, the coefficients $\beta_j$ 
specify the mass prescribed by the target distribution at each sample 
location $y_j$ and encode the structure of the desired reference density.
\end{remark}

\begin{remark}
In this work, the state considered in the optimal transport formulation
corresponds to the spatial position of each agent in $\mathbb{R}^n$. 
Accordingly, both the empirical measure $\mu$ and the target distribution 
$\nu$ are defined over positional variables, and the transport cost 
$\|x - y\|^2$ quantifies displacement in this spatial domain.
\end{remark}

We consider a multi-agent system where each agent follows discrete-time 
dynamics. Initially, agents are located at positions 
\( \{x_i(0)\}_{i=1}^M \), and their spatial configuration is represented 
by the empirical distribution
\(
\mu_0 := \frac{1}{M} \sum_{i=1}^M \delta_{x_i(0)},
\)
where \( \delta_x \) denotes the Dirac measure centered at 
\( x \in \mathbb{R}^n \). At the final time step \( T \), the agents reach 
positions \( \{x_i(T)\}_{i=1}^M \), forming the terminal distribution
\(
\mu_T := \frac{1}{M} \sum_{i=1}^M \delta_{x_i(T)}.
\)
These empirical measures capture the spatial distribution of the agents 
at the start and end of the mission.

In our formulation, the target distribution \( \nu \) is represented in 
a discrete form by a set of sample points \( \{y_j\}_{j=1}^N \), which 
serve as representative support locations of the underlying reference 
density. These points are obtained by discretizing a continuous spatial 
map that specifies the desired target allocation. For instance, in an 
environmental monitoring task, a priority field indicating regions of 
higher sensing importance defines the target density \( \nu \), and the 
resulting sample points \( \{y_j\} \) provide a discrete approximation 
of this field for the transport computation. The associated weights 
\( \beta_j \) encode the mass assigned to each sample point, and this 
discretization is standard in optimal transport when a continuous target 
distribution must be compared against a finite set of agent positions.

Although nonuniform weights \( \beta_j \) may, in general, be assigned 
to the sample points of the target distribution \( \nu \), we adopt the 
uniform choice \( \beta_j = \tfrac{1}{N} \) for simplicity of 
presentation. Under this uniform assumption, the target distribution 
becomes
\(
\nu = \frac{1}{N} \sum_{j=1}^{N} \delta_{y_j}.
\)
This simplification is used solely for clarity in exposition while the 
proposed formulation and control framework apply without modification to 
general nonuniform target weights as well.

The goal is to steer the $M$ agents so that their terminal empirical
distribution $\mu_T = \tfrac{1}{M}\sum_{i=1}^M \delta_{x_i(T)}$
approximates the prescribed target distribution $\nu$ in the
2-Wasserstein metric.  
Equivalently, we aim to determine decentralized control inputs
$\{u_i(k)\}$ such that the final agent positions $x_i(T)$ achieve the
best distribution matching to the $N$ target samples $\{y_j\}$, as
quantified by the Kantorovich formulation in \eqref{eqn:kantorovich-lp}.
This highlights the core problem: guiding the agents’ collective
distribution toward the desired target distribution via an
optimal transport-based control framework.

The agent dynamics are assumed to follow one of the following models:
\begin{enumerate}
    \item \textbf{Linear Time-Invariant (LTI) dynamics:}
    \begin{equation}
        x_i(k+1) = A_i x_i(k) + B_i u_i(k), \label{eqn:LTI}
    \end{equation}
    where $x_i(k) \in \mathbb{R}^n$, $u_i(k) \in \mathbb{R}^m$, and $(A_i, B_i)$ is a controllable pair.

    \item \textbf{Nonlinear control-affine dynamics:}
    \begin{equation}
        x_i(k+1) = f_i(x_i(k)) + g_i(x_i(k)) u_i(k), \label{eqn:nonlinear_con_aff}
    \end{equation}
    where $f_i: \mathbb{R}^n \to \mathbb{R}^n$ and $g_i: \mathbb{R}^n \to \mathbb{R}^{n \times m}$ are smooth maps representing the drift and control influence.
\end{enumerate}

To steer the agents toward a desired distribution, a distributed optimization problem is formulated to align the empirical distribution of the agent states with a prescribed reference distribution. The mismatch is quantified using the squared 2-Wasserstein distance, leading to the optimization problem:
\[
    \min_{\{u_i(k)\}} \; J = \mathcal{W}_2^2\left(\mu_k, \nu\right), \quad \text{subject to} \quad \eqref{eqn:LTI} \text{ or } \eqref{eqn:nonlinear_con_aff},
\]
where \( \mu_k = \frac{1}{M} \sum_{i=1}^M \delta_{x_i(k)} \) is the empirical 
distribution of the agent states at time \( k \), and 
\( \nu = \frac{1}{N}\sum_{j=1}^{N} \delta_{y_j} \) represents the reference 
distribution supported on the target points \( \{y_j\}_{j=1}^N \).
Because each agent state \(x_i(k)\) evolves over time under the system 
dynamics, the empirical distribution \(\mu_k\) is time-varying. As a 
consequence, the objective \( \mathcal{W}_2(\mu_k,\nu) \) is also a 
time-dependent quantity evaluated along the closed-loop evolution of the 
multi-agent system.

This formulation enables the design of distributed control policies that minimize transport cost while respecting system dynamics, supporting decentralized implementation, and unifying both linear and nonlinear agent models under a common optimal transport framework.

\begin{remark}
The optimal control problem becomes challenging once the Wasserstein distance is incorporated into the objective. Although the formulation in \eqref{eqn:kantorovich-lp} is a linear program when the agent positions are fixed, the agent locations in our setting are decision variables that are not known in advance. They are generated by the system dynamics and therefore depend on the control inputs. As a result, the transport cost depends jointly on both the coupling matrix \( \pi_{ij} \) and these control-dependent agent positions, which introduces a bilinear coupling between the two sets of variables. This coupling makes the overall optimization problem nonconvex with respect to the control trajectories even under linear dynamics. In addition, the cost introduces global interactions among all agents, which complicates the design of decentralized control policies under limited communication.
\end{remark}


\section{Approach: Reducing Wasserstein Distance via Sequential Local Updates}

The global optimal transport problem couples all agents and is not directly tractable. To make the assignment manageable, we consider one agent at a time. With the position \(x_i(k)\) fixed, the Kantorovich problem in \eqref{eqn:kantorovich-lp} reduces to determining how this agent distributes its mass \(1/M\) among the sample points.
We call this reduced one-agent problem the \emph{local OT problem}. It has a simple analytic solution: the agent assigns mass to the closest sample points first, subject to the current capacities \( \beta_j \) defined in \eqref{eqn:kantorovich-lp}.

Using this idea, the method consists of two phases: agents first determine their local sample points by sequentially solving their local OT problems, and then each agent independently executes its control inputs to move toward those assigned points.

\subsection{Sample Selection Phase: Sequential Target Assignment}

At the beginning of each \textit{\( H \)-step cycle}, where $H\in\mathbb{N}$, agents sequentially assign local sample points by solving a simplified transport problem. For each agent \( i \in \{1, \dots, M\} \), the following steps are performed:

\begin{enumerate}[leftmargin=*, itemindent=4pt]
\item \textbf{Solve Local Transport Problem:}  
Given a single agent at \( x_i(k) \), the local OT problem 
admits an analytic, greedy solution.  
Define the distances \( d_j := \|x_i(k) - y_j\| \) and reorder the 
indices so that \( d_{j_1} \le d_{j_2} \le \cdots \).  
The agent then allocates portions of its mass \(1/M\) to the nearest 
sample points in this order. Specifically,
$\pi_{i j_\ell}
    = \min\!\left(\beta_{j_\ell},\;
    \frac{1}{M} - \sum_{m<\ell} \pi_{i j_m}\right)$,
where the term 
\(
\frac{1}{M} - \sum_{m<\ell} \pi_{i j_m}
\)
represents the remaining mass after allocations to the closer 
sample points \( j_1,\dots,j_{\ell-1} \).  
This greedy nearest-neighbor allocation continues until the full mass 
\(1/M\) is exhausted and yields a closed-form local transport plan.

\item \textbf{Update Sample Point Weights:}  
Using the computed transport plan \( \pi_{ij} \), the residual 
capacities are updated via  
$\beta_j \leftarrow \beta_j - \pi_{ij}$,
which reduces the available capacity for subsequent agents.

\item \textbf{Assign Local Sample Points:}  
Sample points with \( \pi_{ij} > 0 \) become the local sample points of 
agent \( i \).  
Sequential assignment in a predefined agent order prevents over-allocation 
and ensures that the resulting local sample sets remain capacity-disjoint, 
providing a feasible global transport plan.
\end{enumerate}

Because each agent immediately reduces the residual capacities $\beta_j$
of the sample points it selects, subsequent agents can only allocate mass
to the remaining available capacity. This sequential update prevents
capacity oversubscription and ensures that the local transport plans of
different agents remain capacity-disjoint. In turn, this procedure
produces an aggregate transport plan that is feasible with respect to the
global target distribution, while the full global Kantorovich problem
cannot be solved directly at this stage because the agent positions are not
fixed and would appear as optimization variables. The sequential rule is
more conservative than the global OT assignment, since later
agents must select from a reduced feasible set of sample points. However,
this conservative structure provides a computationally efficient and
provably feasible way to construct a cycle-wise transport plan.

\subsection{Control Execution Phase: Parallel Trajectory Update}

Following sample point assignment, agents enter the control phase, during which each agent independently applies a control law over the next \( H \) steps to minimize the local Wasserstein distance (i.e., the cost induced 
by the agent’s local OT problem) associated with the agent's local sample points. The dynamics follow either the linear system in \eqref{eqn:LTI} or the constrained nonlinear system in \eqref{eqn:nonlinear_con_aff}.

\medskip

This two-phase cycle repeats every \( H \) time steps as follows:
\begin{itemize}[leftmargin=*, itemsep=0pt]
    \item \textbf{Sample Selection Phase:} Agents sequentially solve the local OT problem to determine their respective local sample points.
    \item \textbf{Control Execution Phase:} After selecting local sample points, each agent applies its control inputs over the next \( H \) steps in a distributed manner to minimize the local Wasserstein distance associated with its assigned sample points.
\end{itemize}

Repeating these two phases iteratively updates the agents’ positions in
accordance with the locally assigned sample points. A rigorous analysis
of how these updates influence the global Wasserstein distance is
provided later in Section~VI.

\begin{remark}[Centralized Communication vs. Distributed Control]
The sequential update during Sample Selection Phase requires centralized communication for global coordination, ensuring that agents do not select the same sample points. However, the control input computation can be performed independently by each agent, enabling a \textbf{distributed control} scheme. This allows agents to update their trajectories in parallel after the Sample Selection Phase. The reliance on centralized communication, however, can be impractical in real-world scenarios. Therefore, a more realistic fully decentralized approach, which eliminates the need for centralized coordination, will be presented later.
\end{remark}

In the following sections, we derive the optimal control inputs for both LTI and nonlinear models to minimize the local Wasserstein distance.


\section{Wasserstein Distance Convergence Mapping via Optimal Control}
\subsection{Linear Time-Invariant System Case}
Let agent $i \in \{1, \dots, M\}$ have initial position $x_i(0) \in \mathbb{R}^n$, and assume its dynamics evolve under a discrete-time LTI system \eqref{eqn:LTI}.
The planning horizon is $H \in \mathbb{N}$.

Let the local sample points assigned to agent \( i \) be denoted by \( \{y_j \in \mathbb{R}^n\}_{j \in \mathcal{S}_i} \), where each point \( y_j \) is associated with a non-negative weight \( \beta_{ij} \geq 0 \). The index set \( \mathcal{S}_i \subseteq \{1, \dots, N\} \) represents the subset of global sample points that are locally relevant to agent \( i \), selected based on the process described in Sample Selection Phase. The weights satisfy the mass transport condition \( \sum_{j \in \mathcal{S}_i} \beta_{ij} = \frac{1}{M} \), where \( M \) is the total number of agents.
 The objective is to find the optimal control input for agent $i$ at any time $k$, to minimize the local Wasserstein distance with a horizon length $H$ defined by:
\begin{equation}
J := \mathcal{W}_2^2(\delta_{x_i(k+H)}, \nu_i)=\sum_{j\in\mathcal{S}_i} \tilde{\pi}_{ij} \|x_i(k+H) - y_j\|^2,\label{eqn:obj_LTI}
\end{equation}
where $\tilde{\pi}_{ij}$ denotes the local mass transportation plan from $x_i(k)$ to the local sample points, which is already known from Sample Selection Phase.

Then, the following results are provided to derive the optimal control input $u_i^*$ minimizing \eqref{eqn:obj_LTI}.

\begin{lemma}[Conversion to a Constrained Minimum-Norm Form]
Consider the weighted least-squares problem
\[
\min_{u\in\mathbb{R}^m} 
\sum_{j=1}^p \pi_j \|Ax + Bu - y_j\|^2,
\]
where $\pi_j \ge 0$ and $y^* := (\sum_j \pi_j y_j)/(\sum_j \pi_j)$ is the
weighted average target. If the equation 
$Ax + Bu = y^*$
is feasible, then every solution of this equation achieves the minimum value of the weighted least-squares cost. Among these minimizers, the minimum-norm one is obtained by solving the constrained problem
\[
\min_{u} \|u\|^2 
\quad \text{s.t.} \quad
Ax + Bu = y^*.
\]
This is a standard result for quadratic objectives with linear equality
constraints and hence, the proof is omitted.
\end{lemma}

The previous lemma characterizes the single-step case. In the $H$-step
setting used in the control phase, the control inputs are stacked into
the vector
\begin{equation}
U_i := [u_i(k)^\top,\dots,u_i(k+H-1)^\top]^\top \in \mathbb{R}^{mH},
\label{eq:stacked_input}
\end{equation}
and the resulting state update over $H$ steps is expressed through the
reachability matrix
\begin{equation}
\Phi_{H,i} := [A_i^{H-1}B_i,\; A_i^{H-2}B_i,\; \dots,\; B_i]\in \mathbb{R}^{n\times mH},
\label{eq:reachability_matrix}
\end{equation}
so that the terminal state condition becomes
\begin{equation}
x_i(k+H) = A_i^H x_i(k) + \Phi_{H,i} U_i.
\label{eq:terminal_state_relation}
\end{equation}

The multi-step version of the minimum-norm property is given next.

\begin{lemma}[Minimum-norm solution under affine constraint]
Let $\Phi_{H,i}$ have full row rank with $n \le mH$. For any vectors
$x_i(k)$ and $y_i^{*}$, consider the affine constraint
\begin{equation}
\Phi_{H,i} U_i = y_i^{*} - A_i^{H} x_i(k).  \label{eq:affine_constraint}
\end{equation}
Among all inputs satisfying this constraint, the minimum-norm solution is
$U_i^{*}
= \Phi_{H,i}^{\top}(\Phi_{H,i}\Phi_{H,i}^{\top})^{-1}
(y_i^{*} - A_i^{H}x_i(k))$.

The result follows immediately from the Moore-Penrose pseudoinverse characterization of minimum-norm solutions. Hence, the proof is omitted.
\end{lemma}

\begin{theorem}[Optimal Control for Linear Time-Invariant Systems] 
\label{thm:LTI_control}
Consider the discrete-time LTI system in~\eqref{eqn:LTI}, where the pair
$(A_i,B_i)$ is controllable. Assume the state dimension is $n$ and the
planning horizon satisfies $H \ge n$. The stacked control sequence, the
finite-horizon reachability matrix, and the terminal-state relation are
given in~\eqref{eq:stacked_input}, \eqref{eq:reachability_matrix}, and
\eqref{eq:terminal_state_relation}, respectively.

Let $\mathcal{S}_i \subseteq \{1,\dots,N\}$ denote the local sample
points assigned to agent~$i$, and let $\tilde{\pi}_{ij} \ge 0$ be the
corresponding local transport weights. Define the local barycenter and
total mass by
\begin{equation}
\textstyle y_i^{*} := \frac{1}{\omega_i}\sum_{j\in\mathcal{S}_i}\tilde{\pi}_{ij}y_j,
\qquad
\omega_i := \sum_{j\in\mathcal{S}_i}\tilde{\pi}_{ij}.
\label{eq:centroid}  
\end{equation}
The associated local transport cost is
\begin{equation}
J(U_i) := \sum_{j\in\mathcal{S}_i} \tilde{\pi}_{ij} 
\|x_i(k+H)-y_j\|^2.
\label{eqn:cost_LTI}
\end{equation}

Then, the cost $J(U_i)$ is minimized by the control sequence
\begin{equation}
\begin{aligned}
u_i^{*}(k+t)
= B_i^\top (A_i^\top)^{H-1-t} 
\Gamma_{H,i}^{-1}\bigl(y_i^{*}-A_i^{H}x_i(k)\bigr),\\
\qquad t=0,\dots,H-1,
\end{aligned}
\label{eqn:optimal control for linear}
\end{equation}
where
$\Gamma_{H,i} := \sum_{t=0}^{H-1} A_i^{t}B_iB_i^\top(A_i^\top)^{t}$
is the discrete-time controllability Gramian.
\end{theorem}

\begin{proof}
Since $(A_i,B_i)$ is controllable and $H \ge n$, the reachability matrix
$\Phi_{H,i}$ has full row rank, i.e.,
$\operatorname{range}(\Phi_{H,i})=\mathbb{R}^n$.
Thus, the terminal constraint 
$\Phi_{H,i}U_i = y_i^{*}-A_i^{H}x_i(k)$ 
is feasible for any $y_i^{*}$.
Substituting
\(
x_i(k+H) = A_i^H x_i(k) + \Phi_{H,i}U_i
\)
into \eqref{eqn:cost_LTI} yields the weighted least-squares problem
\begin{equation}
\min_{U_i} 
\sum\nolimits_{j\in\mathcal{S}_i} \tilde{\pi}_{ij}
\bigl\|A_i^H x_i(k) + \Phi_{H,i}U_i - y_j\bigr\|^2.
\label{eqn:least-square}
\end{equation}

By Lemma~1 (weighted least-squares solution), the cost \eqref{eqn:least-square}
is minimized whenever the terminal state equals the weighted barycenter
$y_i^{*}$ in \eqref{eq:centroid}. Hence, every minimizer must satisfy
\[
x_i(k+H) = y_i^{*}
\quad\Longleftrightarrow\quad
\Phi_{H,i} U_i = y_i^{*} - A_i^H x_i(k),
\]
which is feasible by the full-row-rank property of $\Phi_{H,i}$ and matches
the affine constraint in \eqref{eq:affine_constraint}.

Among all $U_i$ satisfying \eqref{eq:affine_constraint}, Lemma~2
(minimum-norm solution under affine constraint) states that the
minimum-norm solution is $U_i^{*}
= \Phi_{H,i}^{\top}(\Phi_{H,i}\Phi_{H,i}^{\top})^{-1}
\bigl(y_i^{*}-A_i^{H}x_i(k)\bigr)$.

For the LTI system in \eqref{eqn:LTI}, one can verify that
\[
\textstyle\Phi_{H,i}\Phi_{H,i}^{\top}
=
\sum_{t=0}^{H-1} A_i^{t}B_iB_i^\top(A_i^\top)^{t}
= \Gamma_{H,i},
\]
and hence, $U_i^{*}
= \Phi_{H,i}^{\top}\Gamma_{H,i}^{-1}
\bigl(y_i^{*}-A_i^{H}x_i(k)\bigr)$.

The $t$th block of $U_i^{*}$ corresponds to $u_i^{*}(k+t)$ and is given by

\vspace{-.15in}
\[
u_i^{*}(k+t)
= B_i^\top (A_i^\top)^{H-1-t} 
\Gamma_{H,i}^{-1}\bigl(y_i^{*}-A_i^{H}x_i(k)\bigr),
\]
for $t=0,\dots,H-1$, which matches \eqref{eqn:optimal control for linear}.
Since every minimizer of \eqref{eqn:least-square} must satisfy the barycenter
condition and $U_i^{*}$ is the unique minimum-norm solution under this
constraint, the sequence $\{u_i^{*}(k+t)\}_{t=0}^{H-1}$ minimizes the cost
$J(U_i)$ in \eqref{eqn:cost_LTI}.
\end{proof}
\vspace{-.15in}

\subsection{Control-Affine Nonlinear System Case}

In the LTI case, the cost function in \eqref{eqn:cost_LTI} leverages the linear system dynamics, with the terminal state \( x_i(k+H) \) as an affine function of the control sequence \( U_i \), resulting in a quadratic cost. This allows the problem to be formulated as a regularized least-squares problem in \eqref{eqn:least-square}, implicitly penalizing control effort without an explicit cost term.

For control-affine nonlinear systems, however, \( x_i(k+H) \) is typically nonlinear in \( U_i \), preventing a quadratic cost form. Without explicit regularization, the optimizer may use large control inputs to reach the target, leading to ill-posed or unstable solutions. To mitigate this, we explicitly add a quadratic penalty on control effort to ensure a well-posed and physically meaningful solution as follows:

\vspace{-.15in}
{\small
\begin{align}
    J(u_i) &= \frac{1}{2} \sum_{j\in\mathcal{S}_i} \tilde{\pi}_{ij} \left\| x_i(k+H) - y_j \right\|^2 
    + \frac{1}{2} \sum_{t=k}^{k+H-1} u_i(t)^\top R u_i(t), \nonumber\\
    &\text{where } \textstyle\tilde{\pi}_{ij} \geq 0 \text{ and }  \sum_{j\in\mathcal{S}_i} \tilde{\pi}_{ij} = \frac{1}{M}. \label{eqn:nonlinear-cost}
\end{align}    
}

The addition of the control cost ensures that the optimization balances the trade-off between achieving the desired transport objective and minimizing control energy.

The optimal control input for the nonlinear system with the cost function \eqref{eqn:nonlinear-cost} is then provided below.

\begin{proposition}[Optimality Conditions for the Nonlinear Case]
For the control-affine dynamics in~\eqref{eqn:nonlinear_con_aff} and the 
finite-horizon cost~\eqref{eqn:nonlinear-cost}, the necessary conditions 
for optimality follow directly from Pontryagin’s minimum principle. The 
resulting state-costate recursion and optimal control law are:
\begin{align}
\lambda_i(t) &= 
\nabla_x f_i(x_i(t))^\top \lambda_i(t+1) \\
&\quad
+ \sum\nolimits_{j=1}^{m} 
    \bigl(\nabla_x g_{i,j}(x_i(t))\,u_{i,j}(t)\bigr)^\top 
    \lambda_i(t+1),\nonumber\\
    u_i(t) &= -R^{-1} g_i(x_i(t))^\top \lambda_i(t+1), \label{eq:non_opt_u}\\
    \lambda_i(k+H) &= 
        \omega_i\bigl(x_i(k+H) - y_i^*\bigr),
\end{align}
where \( \omega_i = \sum_{j\in\mathcal{S}_i} \tilde{\pi}_{ij} \) and 
\( y_i^* \) is the weighted barycenter of the assigned sample points.
These relations form the two-point boundary value system used to compute 
the nonlinear control input.
\end{proposition}

\noindent
The derivation of these optimality conditions is not included, since 
they follow directly from the standard Pontryagin minimum principle for 
discrete-time, control-affine systems. We summarize the result here only 
to state the conditions required for computing the nonlinear controller.

\section{Fully Decentralized Implementation via Local Memory}
The previous framework, particularly the Sample Selection Phase, assumes each agent has access to the target weights of others, requiring global or centralized communication. This is often impractical in real-world scenarios where communication is local and may suffer from packet loss or latency. To overcome this, we propose a fully decentralized scheme where each agent maintains a memory of previously received weights and reuses them when direct communication is unavailable. This allows autonomous weight initialization and updates using only local, asynchronous information. In this way, the memory rule also functions as a simple worst-case fallback mechanism under temporary communication loss, ensuring that the sample-selection step remains well defined even when connectivity is intermittent or time-varying.
\vspace{-.15in}

\subsection{Sequential Update Scheme in a Decentralized Setup}
Suppose \( r \) agents (\( r > 1 \)) are within agent \( i \)'s communication range \( r_c \), forming a local subgroup of \( r+1 \) agents. Within this subgroup, agents sequentially select local sample points similar to the centralized scheme in Section~III-A to avoid barycenter overlap. All-to-all communication within this subgroup is unnecessary since a spanning tree topology suffices to share local sample point selection information.

While sample point overlap with agents outside the communication range may seem concerning, it is typically less impactful due to their spatial separation. (This concern will be clearly handled in the simulation section.) A more significant issue arises from previously connected agents that move out of range but remain nearby. To address this, a local memory mechanism is used to track past selections and minimize conflicts.
\vspace{-.15in}

\subsection{Local Memory Construction}
Let \( \underline{\boldsymbol{\beta}}_j^{(k)} \in \mathbb{R}^{|\mathcal{S}_j|} \) denote the weight vector associated with the
local sample set \( \mathcal{S}_j \) of agent \( j \) at time \( k \), and let \( x_i(k) \in \mathbb{R}^n \) represent the position of agent \( i \). The communication range is defined by a scalar \( r_c > 0 \). Each agent \( i \) maintains a memory structure that stores the most recently received weight vectors from other agents, defined as
\vspace{-.15in}

{\small
\begin{equation}
    \texttt{Memory}_i^{(k)}(j) =
    \begin{cases}
        \boldsymbol{\underline{\beta}}_j^{(k)}, & \text{if } \|x_i(k) - x_j(k)\| < r_c, \\
        \texttt{Memory}_i^{(k-1)}(j), & \text{otherwise}.
    \end{cases}
\end{equation}
}

This rule ensures that memory is updated only when the neighboring agent is within communication range. Otherwise, the agent retains the most recently received weight vector from the previous iteration.
\vspace{-.05in}

\subsection{Weight Update with Memory Decay}

Each agent initializes its weight vector as 
\( \boldsymbol{\underline{\beta}}_{\text{init}} \in \mathbb{R}^N \), 
a global weight vector defined over the entire sample set (distinct from 
\( \boldsymbol{\underline{\beta}}_i^{(k)} \), which applies only to the local sample set 
\( \mathcal{S}_i \)), with uniform entries \( 1/N \). 
Every \( H \) steps, the agent resets its local weight via  
\( \boldsymbol{\underline{\beta}}_i^{(k)} \gets \boldsymbol{\underline{\beta}}_{\text{init}} \).
When agents are within communication range, defined by \( \mathcal{N}_i^{(k)} := \{ j \mid \|x_i(k) - x_j(k)\| < r_c \} \), they form a subgroup and sequentially update their weights to reduce overlap in selecting local sample points. This update, reflected in \( \boldsymbol{\underline{\beta}}_i^{(k)} \), reduces the weights of points already chosen by neighbors, as described in the Sample Selection Phase, which encourages spatial diversity.

To account for previously connected but currently disconnected agents (\( j \notin \mathcal{N}_i^{(k)} \)), the agent applies a memory-based correction:
\(
\boldsymbol{\underline{\beta}}_i^{(k)} \gets \boldsymbol{\underline{\beta}}_i^{(k)} - \gamma \cdot \min_{j \notin \mathcal{N}_i^{(k)}} \texttt{Memory}_i^{(k-1)}(j),
\)
where \( \gamma \in [0,1] \) is a decay factor.

The decay factor \( \gamma \) adjusts how strongly outdated information is suppressed: \( \gamma = 0 \) disables memory correction, while larger values subtract more previously stored information from the current weights, reducing sample-point overlap under intermittent communication.
The element-wise minimum ensures a conservative adjustment, compensating for outdated information while avoiding overrepresentation of already-covered areas. Finally, an element-wise non-negativity projection is applied to ensure \( \boldsymbol{\underline{\beta}}_i^{(k)} \succeq 0 \) as follows:
$\boldsymbol{\underline{\beta}}_i^{(k)} \gets \max\big( \boldsymbol{0}, \boldsymbol{\underline{\beta}}_i^{(k)} \big)$  (element-wise).

The resulting coordination scheme is fully decentralized: each agent
updates its local weights, exchanges information only with neighbors
within its communication range, and maintains a finite memory of past
interactions. The weight update occurs every $H$ steps and requires only
local computations together with a simple reduction rule that resolves
temporary disconnections. Consequently, each agent selects its local
sample points without global synchronization. The procedure remains
operational under intermittent communication and integrates naturally
into the two-stage cycle structure of the algorithm. These structural
properties also facilitate a direct comparison with centralized
counterparts and underpin the convergence analysis developed in the next
section.

\begin{remark}
In a centralized version, the Sample Selection Phase performs a global 
greedy allocation over all $M$ agents and $N$ samples, with complexity 
$O(N \log N + MN)$ and global weight broadcasts each cycle. The proposed 
decentralized scheme replaces this with an $O(N \log N)$ local update 
per agent and neighborhood-only communication, while the control update 
remains decentralized in both cases. Thus, computation and communication 
scale with local neighborhood size rather than the total number of 
agents.
\end{remark}


\section{Cycle-Level Convergence Guarantees}
At the beginning of each cycle, every agent commits to a fixed local transport plan $\tilde{\pi}_{ij}$ determined by the preceding selection phase and follows the corresponding control inputs during the cycle. This section shows that
the resulting evolution possesses a monotonic, cycle-wise descent property and establishes a cycle-level bound on the discrepancy between the agent
distribution and the target distribution. These results provide the key convergence guarantees of the proposed framework.

\begin{theorem}[Cycle-wise Descent of the Surrogate Transport Cost]\label{thm:lyapunov}
Consider a multi-agent system whose states evolve under either the LTI
dynamics~\eqref{eqn:LTI} or the nonlinear dynamics~\eqref{eqn:nonlinear_con_aff}.
For a fixed cycle length $H\in\mathbb{N}$, define the surrogate transport
cost associated with a local transport plan $\tilde{\pi}_{ij}^{(\ell)}\ge0$ by
\[
\textstyle\Psi^{(\ell)}(k)
    := \sum_{i=1}^M \sum_{j\in\mathcal{S}_i^{(\ell)}}
       \tilde{\pi}_{ij}^{(\ell)} \|x_i(k)-y_j\|^2 .
\]

At the beginning of cycle $\ell$, i.e., at $k_\ell=\ell H$, the local
sample-selection phase fixes the sets $\mathcal{S}_i^{(\ell)}$ and selects
nonnegative transport weights $\tilde{\pi}_{ij}^{(\ell)}$ satisfying the
feasibility conditions
\begin{equation}
\textstyle
\sum_{j\in\mathcal{S}_i^{(\ell)}} \tilde{\pi}_{ij}^{(\ell)} = \frac{1}{M},
\qquad
\sum_{i=1}^M \tilde{\pi}_{ij}^{(\ell)} = \beta_j ,\label{eq:OT_feasibility}
\end{equation}
so that $\tilde{\pi}^{(\ell)}$ defines a valid coupling between
$\mu_{k_\ell}$ and $\nu$.  
The corresponding aggregate weights $\omega_i^{(\ell)}$ and barycenters 
$y_i^{*,(\ell)}$ are then computed as in~\eqref{eq:centroid}.

If the control inputs applied over $k\in\{k_\ell,\dots,k_{\ell+1}-1\}$
satisfy the terminal condition
$x_i(k_{\ell+1}) = y_i^{*,(\ell)}$, $i=1,\dots,M$,
then the surrogate cost decreases across the cycle:
\begin{equation}
\Psi^{(\ell)}(k_{\ell+1}) \le \Psi^{(\ell)}(k_\ell),\label{eq:Phi_ineq}
\end{equation}
with strict inequality whenever at least one agent satisfies
$x_i(k_\ell)\neq y_i^{*,(\ell)}$.

Moreover, for any $k\in\{k_\ell,\dots,k_{\ell+1}\}$,
\begin{equation}
\mathcal{W}_2^2(\mu_k,\nu) \le \Psi^{(\ell)}(k). \label{eq:upper_bound}
\end{equation}
\end{theorem}

\begin{proof}
Fix a cycle index $\ell$ and the corresponding local sample set
$\mathcal{S}_i^{(\ell)}$ and weights $\tilde{\pi}_{ij}^{(\ell)}$.
For each agent $i$, define the quadratic function
\begin{equation}
J_i^{(\ell)}(x)
    :=
    \sum\nolimits_{j\in\mathcal{S}_i^{(\ell)}}
    \tilde{\pi}_{ij}^{(\ell)} \|x - y_j\|^2,
    \qquad x\in\mathbb{R}^n.
\label{eq:Ji_def_new}
\end{equation}
Expanding the squared norm gives
\begin{align}
&J_i^{(\ell)}(x)
=
\sum_{j\in\mathcal{S}_i^{(\ell)}}
\tilde{\pi}_{ij}^{(\ell)}
\bigl(\|x\|^2 - 2 x^\top y_j + \|y_j\|^2\bigr)\label{eq:Ji_expanded}\\
&=
\Bigl(\sum_{j\in\mathcal{S}_i^{(\ell)}}\tilde{\pi}_{ij}^{(\ell)}\Bigr)\|x\|^2
 - 2 x^\top
 \sum_{j\in\mathcal{S}_i^{(\ell)}}
 \tilde{\pi}_{ij}^{(\ell)} y_j
 + \sum_{j\in\mathcal{S}_i^{(\ell)}}
   \tilde{\pi}_{ij}^{(\ell)} \|y_j\|^2.\nonumber
\end{align}
Using the definitions of $\omega_i^{(\ell)}$ and $y_i^{*,(\ell)}$ in
\eqref{eq:centroid}, the middle term can be rewritten as
$
- 2 x^\top
 \sum_{j\in\mathcal{S}_i^{(\ell)}}
 \tilde{\pi}_{ij}^{(\ell)} y_j
 = -2 \omega_i^{(\ell)} x^\top y_i^{*,(\ell)}$.

Substituting into~\eqref{eq:Ji_expanded} yields
\begin{align}
J_i^{(\ell)}(x)
&=
\omega_i^{(\ell)}\|x\|^2
 - 2 \omega_i^{(\ell)} x^\top y_i^{*,(\ell)}
 + \sum\nolimits_{j\in\mathcal{S}_i^{(\ell)}}
   \tilde{\pi}_{ij}^{(\ell)} \|y_j\|^2\nonumber\\
&=
\omega_i^{(\ell)}\|x-y_i^{*,(\ell)}\|^2
 + C_i^{(\ell)},
\label{eq:Ji_completed_square_new}
\end{align}
where $C_i^{(\ell)}
    :=
    \sum_{j\in\mathcal{S}_i^{(\ell)}}
    \tilde{\pi}_{ij}^{(\ell)} \|y_j\|^2
    - \omega_i^{(\ell)} \|y_i^{*,(\ell)}\|^2$
depends only on the sample locations and weights, and is independent of~$x$.

Since $\omega_i^{(\ell)} > 0$ by construction, the term
$\omega_i^{(\ell)}\|x-y_i^{*,(\ell)}\|^2$ is nonnegative and vanishes
if and only if $x=y_i^{*,(\ell)}$. Thus $J_i^{(\ell)}$ is a strictly
convex quadratic function with unique minimizer $y_i^{*,(\ell)}$.

Evaluating~\eqref{eq:Ji_completed_square_new} at the beginning and the end
of the cycle, we obtain
\vspace{-.25in}

\begin{align*}
J_i^{(\ell)}(x_i(k_\ell))
&=
\omega_i^{(\ell)}
\|x_i(k_\ell) - y_i^{*,(\ell)}\|^2
+ C_i^{(\ell)},\\
J_i^{(\ell)}(x_i(k_{\ell+1}))
&=
\omega_i^{(\ell)}
\|x_i(k_{\ell+1}) - y_i^{*,(\ell)}\|^2
+ C_i^{(\ell)}.
\end{align*}
The terminal condition
$x_i(k_{\ell+1}) = y_i^{*,(\ell)}$ for all $i$
implies
\begin{equation*}
    \begin{aligned}
        J_i^{(\ell)}(x_i(k_{\ell+1}))
        = C_i^{(\ell)}
        &\le
        \omega_i^{(\ell)}
        \|x_i(k_\ell) - y_i^{*,(\ell)}\|^2
        + C_i^{(\ell)}\\
        &= J_i^{(\ell)}(x_i(k_\ell)),
    \end{aligned}
\end{equation*}
with strict inequality whenever $x_i(k_\ell)\neq y_i^{*,(\ell)}$.

Summing over $i$ and using the definition of $\Psi^{(\ell)}$ gives
{\small
\begin{equation*}
\Psi^{(\ell)}(k_{\ell+1})
=
\sum_{i=1}^M J_i^{(\ell)}(x_i(k_{\ell+1}))
\le
\sum_{i=1}^M J_i^{(\ell)}(x_i(k_\ell))
=
\Psi^{(\ell)}(k_\ell),
\end{equation*}
}

\noindent and strict descent when at least one agent moves.

We now prove~\eqref{eq:upper_bound}.
For any $k\in\{k_\ell,\dots,k_{\ell+1}\}$, the empirical distribution
is
$\mu_k = \frac{1}{M}\sum_{i=1}^M \delta_{x_i(k)}$.
By construction of the local transport weights at the beginning of
cycle $\ell$, the entries $\tilde{\pi}_{ij}^{(\ell)}$ satisfy
$\sum_{j=1}^N \tilde{\pi}_{ij}^{(\ell)} = \frac{1}{M}$,
$\sum_{i=1}^M \tilde{\pi}_{ij}^{(\ell)} = \beta_j$,
for all $i$ and $j$. Because the weights are kept fixed throughout the
cycle, these equalities remain valid for every $k\in\{k_\ell,\dots,k_{\ell+1}\}$.
Hence, $\tilde{\pi}^{(\ell)}$ defines a feasible coupling between
$\mu_k$ and $\nu$ for all such $k$, and the corresponding transport cost
is exactly
$
\sum_{i=1}^M \sum_{j=1}^N
   \tilde{\pi}_{ij}^{(\ell)} \|x_i(k) - y_j\|^2
= \Psi^{(\ell)}(k).
$
Since $\mathcal{W}_2^2(\mu_k,\nu)$ is the \textit{minimum transport cost} over
all feasible couplings, it follows that
$
\mathcal{W}_2^2(\mu_k,\nu)
\;\le\;
\Psi^{(\ell)}(k).
$
\end{proof}

\begin{remark}\label{remark:5}
The inequality \eqref{eq:upper_bound} holds whenever the cycle-wise transport plan $\tilde{\pi}_{ij}^{(\ell)}$ satisfies the usual OT feasibility condition \eqref{eq:OT_feasibility}.
This condition is automatically enforced in the centralized implementation.
In decentralized variants, however, local updates do not necessarily preserve the global
weight balance, and the inequality may therefore fail to hold at every cycle.
Nevertheless, the cycle-wise decrease of \( \Psi^{(\ell)} \) in \eqref{eq:Phi_ineq} always holds in all cases (centralized, decentralized, with or without memory) because the descent argument depends only on the agents moving to their assigned barycenters.
\end{remark}
\vspace{-.1in}

\section{Simulations}
To validate the technical soundness of the proposed method, we conducted a series of simulations on both LTI and nonlinear systems. The simulation settings and results are detailed below, along with quantitative evaluations using the Wasserstein distance as the performance metric.

\vspace{-0.075in}
\subsection{LTI System Case}
An LTI system is considered with synthetically generated system matrices $A_i \in \mathbb{R}^{2 \times 2}$ and $B_i \in \mathbb{R}^2$, which are identical across all 30 agents. Further, the planning horizon is set to $H=50$. The simulation results are shown in Fig.~\ref{fig:sim-LTI}, where the target distribution is illustrated by 1,000 green dots. 
The agents' initial and final positions are marked by blue crosses and red squares, respectively, and their trajectories are depicted by black lines. 
Since the target distribution is a multimodal Gaussian mixture, the induced Wasserstein cost landscape is highly nonconvex and contains 
multiple attraction regions around its modes. Under the decentralized, memoryless policy, agents initialized near a particular mode tend to converge toward that mode and may remain concentrated there.

Two scenarios are considered: Fig.~\ref{fig:sim-LTI}(a) shows a concentrated initial distribution $x(0)$, uniformly sampled from a small square while Fig.~\ref{fig:sim-LTI}(b) displays a uniform distribution across the entire domain.

\begin{figure}[!t]
    \centering
    \subfloat[centralized, concentrated $x(0)$]{
    \includegraphics[width=0.47\linewidth, trim=100 220 120 250, clip]{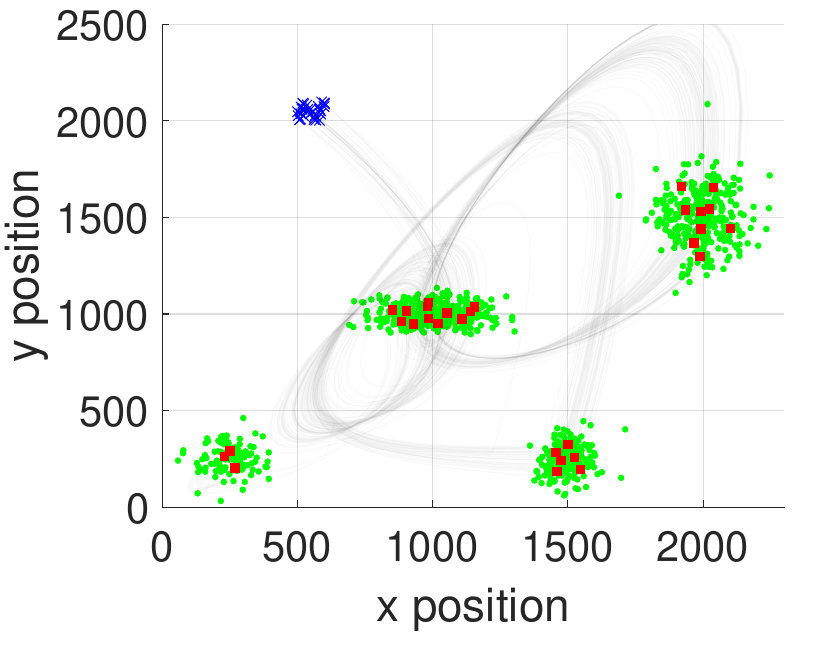}}\;
    \subfloat[decentralized w/o memory ($\gamma=0$), uniform $x(0)$]{
    \includegraphics[width=0.47\linewidth, trim=100 220 120 250, clip]{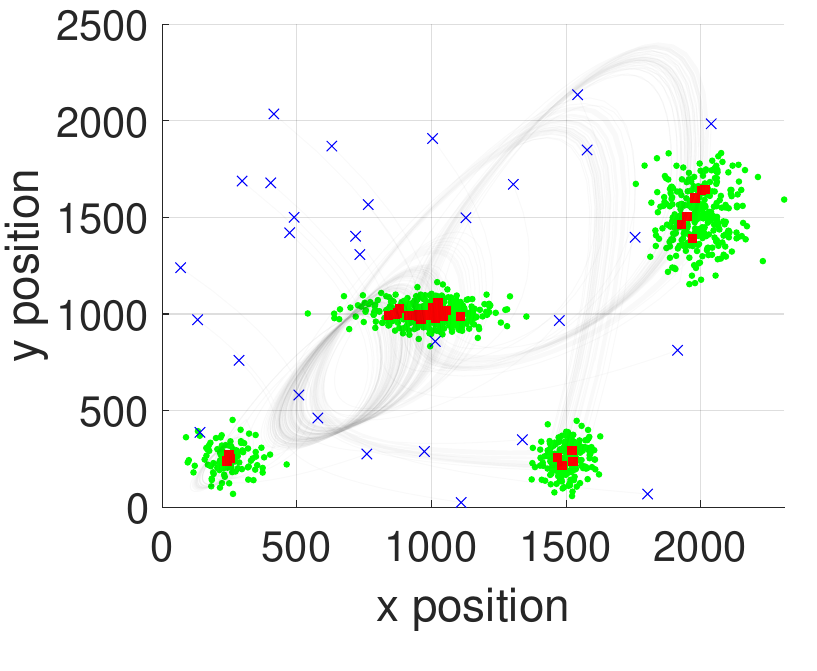}}
    \caption{Distribution matching for the LTI system}
    \label{fig:sim-LTI}
\end{figure}

In Fig.~\ref{fig:sim-LTI}(a), the centralized communication with the distributed controller enables successful distribution matching, even under biased initial conditions, by sharing target-point weights globally. Fig.~\ref{fig:sim-LTI}(b) evaluates a decentralized scheme with a limited communication range ($r_c = 20$) and no memory ($\gamma = 0$), where agents still approximate the target distribution reasonably well due to the wide initial spread even without memory.
\vspace{-.1in}

\begin{figure}[!b]
    \centering
    \subfloat[Decentralized, w/o memory ($\gamma=0$)]{
    \includegraphics[width=0.46\linewidth, trim=130 220 150 250, clip]{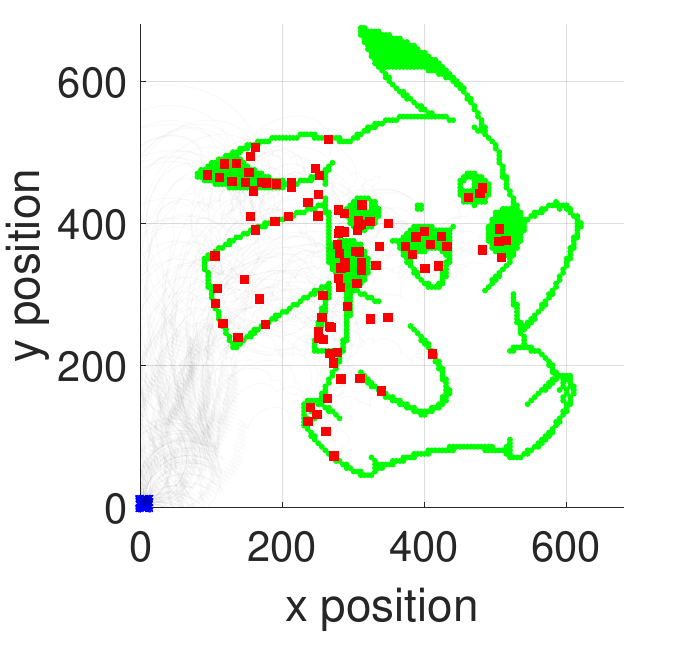}}\quad
    \subfloat[Decentralized, w/ memory ($\gamma=0.7$)]{
    \includegraphics[width=0.46\linewidth, trim=130 220 150 250, clip]{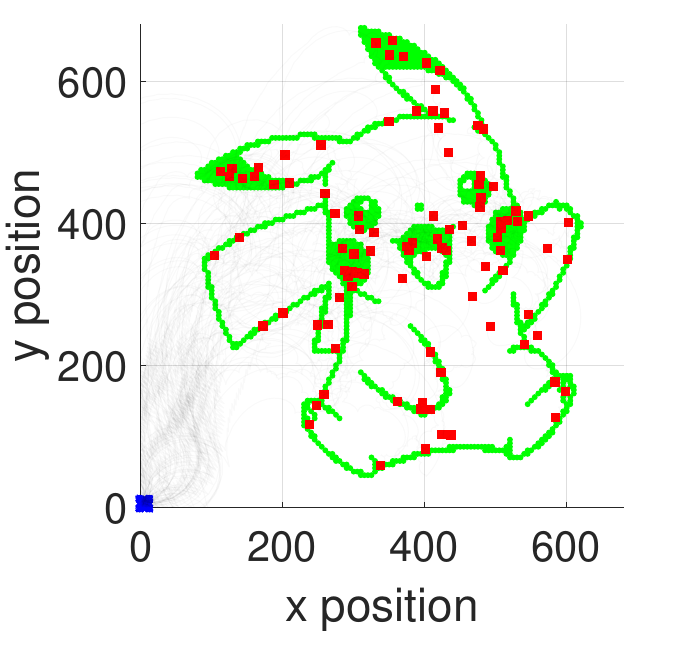}}
    \caption{Distribution matching for the unicycle model.}
    \label{fig:sim-non}
\end{figure}
\subsection{Nonlinear System Case}
We consider a control-affine nonlinear system modeled by unicycle dynamics. 
While the optimality conditions in~\eqref{eq:non_opt_u} characterize the $H$-step solution, 
in our simulations we adopt a one-step receding-horizon implementation for computational simplicity. 
This yields a closed-form control update that is applied repeatedly over a scheduling cycle of $H=50$ steps, 
during which the assigned barycenter remains fixed.

In this setup, 100 agents are initialized near the origin and tasked with matching a non-uniform target distribution, shown in green in Fig. \ref{fig:sim-non}. The target distribution is represented by 1,538 sample points while the same symbol convention is used as in the LTI case.

Fig.~\ref{fig:sim-non}(a) shows that, without memory, decentralized agents with \( r_c = 20 \) struggle to match the target distribution due to limited coordination. In contrast, Fig.~\ref{fig:sim-non}(b) demonstrates that incorporating memory significantly improves distribution matching, as agents implicitly coordinate their behaviors.
\vspace{-.1in}

\subsection{Cycle-Wise Convergence Behavior}

We evaluate the cycle-wise behavior of the surrogate transport cost and the
true Wasserstein distance using the same nonlinear scenario shown earlier in
Fig.~\ref{fig:sim-non}. Each cycle consists of $H=50$ steps, and the results in
Fig.~\ref{fig:W-dist} report the values only at the cycle boundaries
$k_\ell=\ell H$ for 20 cycles.

\begin{figure}[!t]
    \centering
    \subfloat[Centralized]{
    \includegraphics[width=0.46\linewidth]{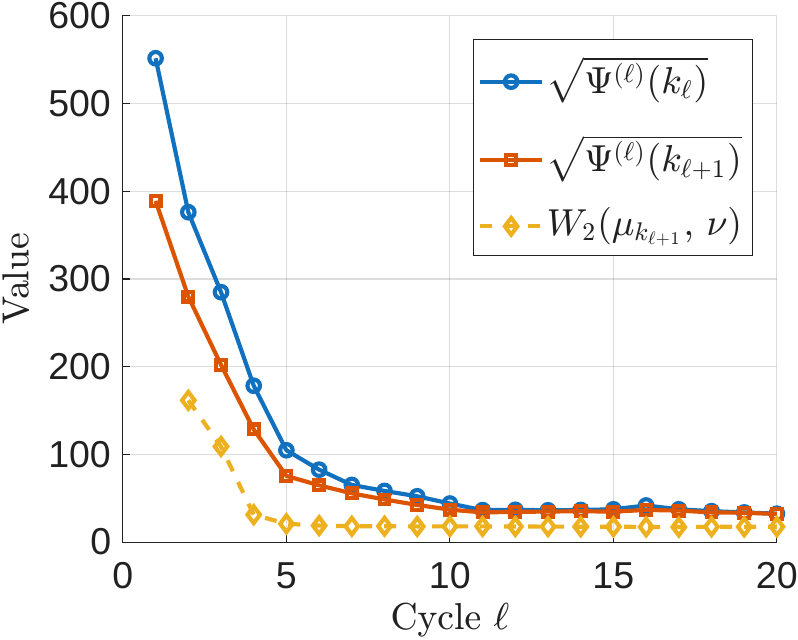}}\quad
    \subfloat[Decentralized, w/ memory ($\gamma=0.7$)]{
    \includegraphics[width=0.46\linewidth]{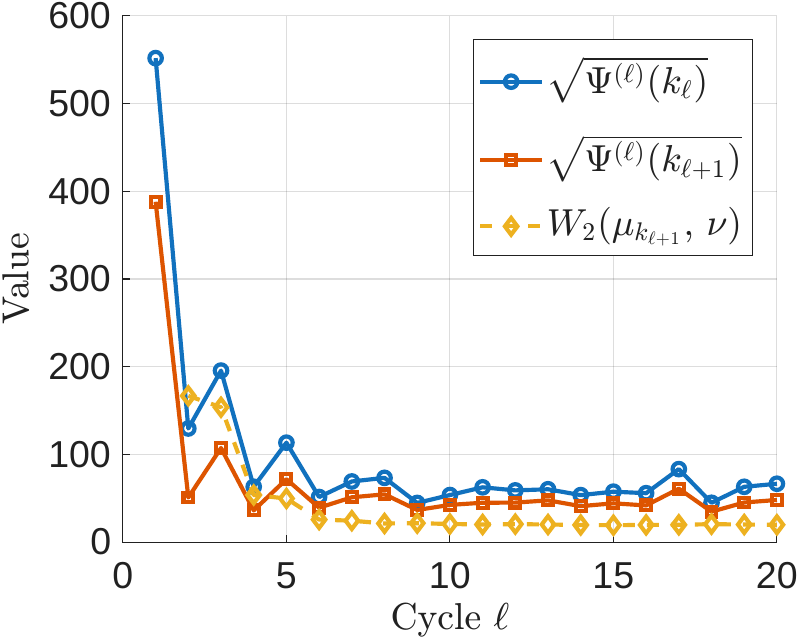}}
    \caption{Cycle-wise evolution of $\sqrt{\Psi^{(\ell)}(k_{\ell})}$, $\sqrt{\Psi^{(\ell)}(k_{\ell+1})}$, and
    $\mathcal{W}_2(\mu_{k_{\ell+1}},\nu)$.}
    \label{fig:W-dist}
    \vspace{-.1in}
\end{figure}

In the centralized case (Fig.~\ref{fig:W-dist}(a)), the simulation behavior
matches the theoretical prediction precisely. The surrogate transport cost at
the end of each cycle (red squares) is always less than or equal to its value
at the beginning of that cycle (blue circles), exactly as stated in Theorem~2.
Furthermore, the Wasserstein distance (yellow diamonds) consistently remains
below the red curve at every cycle boundary, confirming that the surrogate cost
acts as a valid upper bound under centralized assignment. This confirms that the
centrally computed transport plan always satisfies the global feasibility
conditions required by the theory, ensuring that the surrogate cost is a valid
upper bound on the true transport distance.

The decentralized case with memory (Fig.~\ref{fig:W-dist}(b)) exhibits a
different pattern. The surrogate cost at the end of each cycle (red squares)
remains less than or equal to its value at the beginning of that cycle (blue
circles), indicating that agents consistently reach their assigned barycenters
despite limited communication. However, the Wasserstein distance (yellow
diamonds) does not necessarily stay below the surrogate cost. This discrepancy
occurs because, without global synchronization, different agents may update the
residual target weights using only local and potentially outdated information.
Such inconsistencies lead to duplicated reductions of the same target weights,
breaking the global feasibility conditions \eqref{eq:OT_feasibility} required for the surrogate cost to
serve as an upper bound. Nonetheless, the descending red curve shows that the method still makes steady
cycle-wise progress toward the target distribution as described in Theorem \ref{thm:lyapunov} and Remark \ref{remark:5}, while the theoretical
guarantee available in the centralized setting no longer applies.

\section{Conclusion}

This paper presented a decentralized control framework for terminal distribution matching in multi-agent systems with discrete-time LTI and nonlinear dynamics. The method leverages optimal transport theory and the Wasserstein distance to guide agents toward a prescribed final distribution. A sequential update strategy enables communication-aware coordination, and a memory-based weight update scheme provides robustness under intermittent connectivity by reusing decayed information from previously connected neighbors. The framework admits a cycle-wise convergence guarantee under the required feasibility condition associated with the transport weights. Simulation results confirmed accurate terminal distribution matching with low overlap and reliable performance under limited communications.


\bibliographystyle{IEEEtran}
\bibliography{references}

\end{document}